\documentclass[submission,copyright,creativecommons]{eptcs}
\providecommand{\event}{Non-Classical Logics 2024 (NCL'24)}

\usepackage{xifthen}

\ifpdf
  \usepackage{underscore}         
  \usepackage[T1]{fontenc}        
\else
  \usepackage{breakurl}           
\fi

\title{Complexity of Nonassociative Lambek Calculus with classical logic}
\author{Pawe\l{} P\l{}aczek
\institute{WSB Merito University in Poznan, Poland}
\email{pawel.placzek@poznan.merito.pl}
}
\def\titlerunning{Complexity of NL with CPL}
\def\authorrunning{P. P\l{}aczek}

\usepackage{amsmath,amssymb,amsthm}
\usepackage{cleveref}
\usepackage[mathscr]{euscript}
\usepackage{xspace}
\usepackage{enumerate}
\usepackage{cmll}
\usepackage{array,multirow}
\usepackage{hhline}

\renewcommand{\implies}[2]{\ifthenelse{\isempty{#1}}{}{\text{if~}#1\text{,}}\ifthenelse{\isempty{#2}}{}{\text{~then~}#2}}
\renewcommand{\land}{\text{~and~}}
\renewcommand{\lor}{\text{~or~}}
\newcommand\powerset[1]{\mathscr{P}(#1)}

\let\t\forall
\renewcommand\forall{\operatorname{\t}\limits\limits}
\let\tt\exists
\renewcommand\exists{\operatorname{\tt}\limits\limits}

\newcommand{\incircbin}[1]{#1'}

\newcommand{\owedge}{\incircbin{\wedge}}
\newcommand{\ovee}{\incircbin{\vee}}
\newcommand{\ocirc}{\incircbin{\otimes}}
\newcommand{\oneg}{\incircbin{\neg}}
\newcommand{\otop}{\incircbin{\top}}
\newcommand{\obot}{\incircbin{\bot}}
\newcommand{\omultimap}{\incircbin{\multimap}}
\newcommand{\omultimapinv}{\incircbin{\multimapinv}}

\newcommand{\oleq}{\incircbin{\leq}}
\newcommand{\onleq}{\incircbin{\not\leq}}
\newcommand{\oi}{\incircbin{1}}

\theoremstyle{definition}
\newtheorem{definition}{Definition}[section]

\theoremstyle{plain}
\newtheorem{theorem}[definition]{Theorem}
\newtheorem{lemma}[definition]{Lemma}
\newtheorem{corollary}[definition]{Corollary}
\newtheorem{proposition}[definition]{Proposition}

\theoremstyle{remark}

\newcolumntype{C}{>{\hfil\displaystyle}c<{\hfil}}
\newcolumntype{L}{>{\displaystyle}l<{\hfil}}

\newenvironment{marray}{\[\array{lcc}}{\endarray\]\ignorespacesafterend}
\newenvironment{rules}{\[\array{c@{}L@{}c@{}L}\hspace*{0.1\textwidth} & \hspace*{0.4\textwidth} & \hspace*{0.1\textwidth} & \hspace*{0.4\textwidth}\\[-3em]}{\endarray\]\ignorespacesafterend}

\newcommand\EXPTIME{\textit{EXPTIME}\xspace}
\newcommand\PTIME{\textit{PTIME}\xspace}

\usepackage{cite}

\renewcommand{\multimap}{\backslash}
\renewcommand{\multimapinv}{/}

\begin{document}
\maketitle

\begin{abstract}
The Nonassociative Lambek Calculus (NL) represents a logic devoid of the structural rules of exchange, weakening, and contraction, and it does not presume the associativity of its connectives. Its finitary consequence relation is decidable in polynomial time. However, the addition of classical connectives conjunction and disjunction (FNL) makes the consequence relation undecidable. Interestingly, if these connectives are distributive, the consequence relation is decidable in exponential time. This paper provides the proof, that we can merge classical logic and NL (i.e. BFNL), and still the consequence relation is decidable in exponential time.
\end{abstract}

\renewcommand{\arraystretch}{2}

\section{Introduction and preliminaries}

Lambek Calculus L was introduced by Lambek \cite{lambekL} under the name \textit{Syntactic Calculus}. L is a propositional logic with three connectives \(\otimes\) (product), \(\multimap\) and \(\multimapinv\) (residuations of product). Lambek \cite{lambekCUT} introduced the nonassociative version of this logic, nowadays called Nonassociative Lambek Calculus (NL). From a logical perspective, NL can be seen as the pure logic of residuation, and L as its stronger version for associative product. For both L and NL, J. Lambek provided a sequent system and proved cut elimination \cite{lambekL,lambekCUT}.

The product for both L and NL derives from conjunction after dropping the structural rules of exchange, weakening, and contraction in terms of sequent systems. NL additionally does not require being an associative operator in terms of algebra. In effect, we obtain a pure operation joining two formulas. This operation may be seen as a binary modality.

\begin{definition}
Let \(\mathbf{G} = (G, \otimes, \multimap, \multimapinv, \leq)\) be a structure such that \((G, \otimes)\) is a groupoid, \((G, \leq)\) is a poset, and the following holds:
\begin{marray}
\text{(RES)} & a \otimes b \leq c \text{ iff } b \leq a \multimap c \text{ iff } a \leq c \multimapinv b
\end{marray}
for all \(a, b, c \in G\). Then \(\mathbf{G}\) is called a \emph{residuated groupoid}.
\end{definition}

By \emph{groupoid} we mean a set closed under a binary operation without any specific properties required. The residuated groupoids are models of NL. The residuated groupoids where the product is associative are called \emph{residuated semigroups} and are models of L.

The most popular extensions of L and NL are: adding a constant 1 or adding conjunction and disjunction. The constant 1 in algebras is a unit for the product. The conjunction and disjunction replace the partial order with tshe lattice structure and lattice order. We can also add the boundaries, i.e., \(\top\) and \(\bot\), as respectively, the greatest and lowest elements. In this paper we use the same symbol for both syntactic and semantic purposes and the exact meaning is clear from the context.

\begin{definition}
Let \((G, \otimes, \multimap, \multimapinv, \leq)\) be a residuated groupoid and let \(1 \in G\) be an element such that:
\[
1 \otimes a = a = a \otimes 1
\]
for all \(a \in G\). Then \((G, \otimes, \multimap, \multimapinv, 1, \leq)\) is a unital residuated groupoid.
\end{definition}

The unital residuated groupoids are models for NL with constant 1 and unital residuated semigroups are models for L with constant 1.

Lambek Calculus with additive connectives (conjunction and disjunction) is called Full Lambek Calculus and denoted FL. Some authors also require the presence of 1 (multiplicative constant) and \(\top, \bot\) (additive constants). In this paper, we follow this convention, so FL admits all these constants. Analogously, FNL is an extension of NL with additive connectives and all constants.

\begin{definition}
Let \((G, \otimes, \multimap, \multimapinv, 1, \leq)\) be a unital residuated groupoid and \((G, \vee, \wedge, \top, \bot, \leq)\) be a bounded lattice. Then, \((G, \otimes, \multimap, \multimapinv, \vee, \wedge, 1, \top, \bot, \leq)\) is a \emph{residuated lattice}.
\end{definition}The residuated lattices are models for FNL. Residuated lattices where \(\otimes\) is associative are models for FL.

Pentus \cite{pentus} proves that pure L is NP-complete and Buszkowski \cite{buszkoL} proves that its finitary consequence relation is undecidable. A similar situation applies if we add the constant 1. FL is a strongly conservative extension\footnote{A logic \(\mathcal{L}_2\), extending \(\mathcal{L}_1\), is a (resp. strongly) conservative extension of \(\mathcal{L}_1\), if both logics have the same theorems (resp. the same consequence relation) in language of \(\mathcal{L}_1\)} of L, so its finitary consequence relation is also undecidable. The same applies to all strongly conservative extensions of L. In this paper, we focus on extensions of NL because of that.

Buszkowski \cite{buszkoL} proves that the finitary consequence relation for NL is in \PTIME. The same applies if we admit the multiplicative constant. Unfortunately, FNL has an undecidable consequence relation \cite{chvalovsky}.

The lattices in the algebras of FNL are not necessarily distributive. If we consider logic with such an axiom for additive connectives, we talk about Distributive Full Nonassociative Lambek Calculus and denote it DFNL. The models for this logic are residuated distributive lattices.

The finitary consequence relation of DFNL is \EXPTIME-complete if we do not admit the multiplicative constant 1 and is in \EXPTIME if we admit the constant, which was proved in \cite{shkatov}.\footnote{Shkatov and Van Alten \cite{shkatov} show that the satisfiability problem of quantifier-free first-order formulas in the language of bounded distributive residuated lattices is \EXPTIME-complete.} The lower bound of complexity of the consequence relation for DFNL with constant 1 remains an open problem.

The other interesting extensions of FNL are BFNL and HFNL, i.e., Boolean FNL and Heyting FNL. These logics may be seen as extensions of NL with Boolean and Heyting algebras or as extensions of classical logic and intuitionistic logic with NL. Such logics have been studied by Galatos and Jipsen \cite{galatos2017}, Buszkowski \cite{buszko2021}, and others.

\begin{definition}
Let \((G, \otimes, \multimap, \multimapinv, 1, \leq)\) be a unital residuated groupoid and \((G, \vee, \wedge, \neg, \bot, \top, \leq)\) be a Boolean algebra. Then, \((G, \otimes, \multimap, \multimapinv, \vee, \wedge, \neg, 1, \top, \bot, \leq)\) is a \emph{residuated Boolean algebra}.
\end{definition}

In this paper, we provide the proof of the upper bound of the complexity of the consequence relation for BFNL, extending the results of \cite{shkatov}, using the same methods. We also use the results from \cite{vanalten}, where distributive lattices, Heyting algebras, and Boolean algebras are considered. The differences between \cite{shkatov,vanalten} and this paper lay in the details. An experienced reader can easily deduce the results of this paper by reading cited papers, but some changes are subtle, e.g. in some places we do not use families of upsets, but the whole powerset, because we have negation here. Moreover, the results in \cite{shkatov,vanalten} are described in only algebraic terms and use first-order formulas. Here, we use syntactic notion more directly, still using algebraic methods in proofs.

We show the full proof only for the version with the constant 1 because the proofs for logics without that constant can be easily obtained by omitting some parts.

The proof for HFNL may be done analogously. It is necessary to adjust some definitions and conditions, but the idea remains the same.

Since HFNL and BFNL without 1 are strongly conservative extensions of DFNL,\footnote{See Remark 5 in \cite{buszko2021}.} we know their finitary consequence relations are \EXPTIME-hard and, in effect, are \EXPTIME-complete. The lower bound for HFNL and BFNL with 1 is still an open problem.

In the second section, we provide the sequent system for BFNL. This system comes from \cite{galatos2017}, where the authors prove the cut-elimination theorem. In the third section, we study partial structures connected with models of BFNL. We prove important theorems that allow us to check whether a given partial structure is a partial residuated algebra. In the last section, we use these theorems to prove \EXPTIME complexity of the consequence relation for BFNL.

\section{Sequent system}

The language of BFNL is defined as follows. We admit a countable set of variables, which we denote by small Latin letters. The formulas are constructed from this set of variables by five binary connectives (\(\otimes, \multimap, \multimapinv, \vee, \wedge\)), one unary connective (\(\neg\)) and three constants (\(1, \top, \bot\)).

Usual notion of sequents using sequents of formulas is not applicable in nonassociative framework. The comma in sequences is a concatenation operation which is associative. We need to change the structure to something more flexible. Moreover, we need to have two types of commas: one for \(\otimes\) and one for \(\wedge\) with different properites.

We define bunches. The bunches are elements of free biunital bigroupoid, i.e. the algebra with two binary operations with a unit for both of them, generated from the set of all formulas. We denote first operator by comma and the second one by semicolon. The unit for comma is denoted \(\epsilon\) and unit for semicolon is \(\delta\).

One may think of bunches as of binary trees in which leaves are formulas or \(\epsilon\) or \(\delta\) and every node besides leaves is labeled by comma or semicolon.

The bunch \(\epsilon\) is called an \emph{empty bunch}. All the other bunches are nonempty. We reserve Latin capital letters for formulas and Greek capital letters for bunches. A \emph{context} is a bunch with an anonymous variable. Contexts are denoted by \(\Gamma[\_]\), and when we perform the substitution of \(\Delta\) in place of \(\_\), we represent it as \(\Gamma[\Delta]\).

A \emph{sequent} is a pair \(\Gamma\), \(A\), where \(\Gamma\) is a bunch and \(A\) is a formula. We write \(\Gamma \Rightarrow A\).

The axioms and the rules for BFNL are as follows:
\begin{rules}
\text{(id)} & A \Rightarrow A & \text{(cut)} & \frac{\Gamma \Rightarrow A \quad \Delta[A] \Rightarrow C}{\Delta[\Gamma] \Rightarrow C}\\
(\otimes\Rightarrow) & \frac{\Gamma[(A, B)] \Rightarrow C}{\Gamma[A \otimes B] \Rightarrow C} & (\Rightarrow\otimes) & \frac{\Gamma \Rightarrow A \quad \Delta \Rightarrow B}{\Gamma, \Delta \Rightarrow A \otimes B}\\
(\multimap\Rightarrow)& \frac{\Gamma[B] \Rightarrow C \quad \Theta \Rightarrow A}{\Gamma[(\Theta, A \multimap B)] \Rightarrow C} & (\Rightarrow\multimap) & \frac{A, \Gamma \Rightarrow B}{\Gamma \Rightarrow A \multimap B}\\
(\multimapinv\Rightarrow) & \frac{\Gamma[A] \Rightarrow C \quad \Theta \Rightarrow B}{\Gamma[(A \multimapinv B, \Theta)] \Rightarrow C} & (\Rightarrow\multimapinv) & \frac{\Gamma, B \Rightarrow A}{\Gamma \Rightarrow A \multimapinv B}\\
\end{rules}
\begin{rules}
(\wedge\Rightarrow) & \frac{\Gamma[(A; B)] \Rightarrow C}{\Gamma[A \wedge B] \Rightarrow C} & (\Rightarrow\wedge) & \frac{\Gamma \Rightarrow A \quad \Gamma \Rightarrow B}{\Gamma \Rightarrow A \wedge B}\\
(\vee\Rightarrow) & \frac{\Gamma[A] \Rightarrow C \quad \Gamma[B] \Rightarrow C}{\Gamma[A \vee B] \Rightarrow C} & (\Rightarrow\vee) & \frac{\Gamma \Rightarrow A}{\Gamma \Rightarrow A \vee B}\quad \frac{\Gamma \Rightarrow B}{\Gamma \Rightarrow A \vee B}\\
(\top\Rightarrow) & \frac{\Gamma[\Delta] \Rightarrow C}{\Gamma[(\top ; \Delta)] \Rightarrow C} \, \frac{\Gamma[\Delta] \Rightarrow C}{\Gamma[(\Delta ; \top)] \Rightarrow C} & (\Rightarrow\top) & \Gamma \Rightarrow \top\\
(\bot\Rightarrow) & \Gamma[\bot] \Rightarrow C\\
\end{rules}
\begin{rules}
\text{(\(\wedge\)-ass)} & \frac{\Gamma[\Delta_1 ; (\Delta_2 ; \Delta_3)] \Rightarrow C}{\overline{\Gamma[(\Delta_1 ; \Delta_2) ; \Delta_3] \Rightarrow C}} & \text{(\(\wedge\)-ex)} & \frac{\Gamma[\Delta ; \Theta] \Rightarrow C}{\Gamma[ \Theta ; \Delta ] \Rightarrow C}\\
\text{(\(\wedge\)-weak)} & \frac{\Gamma[\Delta] \Rightarrow C}{\Gamma[\Delta ; \Theta] \Rightarrow C} & \text{(\(\wedge\)-cont)} & \frac{\Gamma[\Delta ; \Delta] \Rightarrow C}{\Gamma[\Delta] \Rightarrow C}\\
\end{rules}
\begin{rules}
(\neg\Rightarrow) & A \wedge \neg A \Rightarrow \bot & (\Rightarrow\neg) & \top \Rightarrow A \vee \neg A\\
(1\Rightarrow) & \frac{\Gamma[\Delta] \Rightarrow C}{\Gamma[(1, \Delta)] \Rightarrow C} \quad \frac{\Gamma[\Delta] \Rightarrow C}{\Gamma[(\Delta, 1)] \Rightarrow C} & (\Rightarrow 1) & \epsilon \Rightarrow 1
\end{rules}

We shortly describe the semantics of BFNL. The models for BNFL are residuated Boolean algebras. The valuation is a homomorphism \(\mu\) from the free algebra of formulas to a residuated Boolean algebra \(\mathbf{B}\) extended to bunches inductively as follows:
\begin{gather*}
\mu(\epsilon) = 1\\
\mu(\delta) = \top\\
\mu((\Gamma, \Delta)) = \mu(\Gamma) \otimes \mu(\Delta)\\
\mu((\Gamma; \Delta)) = \mu(\Gamma) \wedge \mu(\Delta)
\end{gather*}

The sequent \(\Gamma \Rightarrow A\) is said to be true in \(\mathbf{B}\) under the valuation \(\mu\) if \(\mu(\Gamma) \leq \mu(A)\).

\section{Partial residuated Boolean algebras}

In this section we provide the notion of partial structures and we prove some properties. The most important result here is \Cref{partial-B1} which helps in identifying partial residuated Boolean algebras in exponential time in the next section.

\subsection{Partial structures}

\begin{definition}
A function \(f: U \mapsto Y\), where \(U \subseteq X\), is called a \emph{partial function} from \(X\) to \(Y\) (we write \(f: X \to Y\)). If \(U = X\), then the function is said to be \emph{total}.
\end{definition}

We write \(f(x) = \infty\), if the function \(f\) on the argument \(x\) is undefined.

\begin{definition}
Let \(I, J, K\) be finite indexing sets. We say \((U, \{f^{n_i}_i\}_{i \in I}, \{a_j\}_{j \in J}, \{R^{m_k}_k\}_{k \in K})\) is a \emph{partial structure}, if \(\{a_j\}_{j \in J} \subseteq U\) and \(f^{n_i}_i: U^{n_i} \to U\) is a partial function for all \(i \in I\) and \(R^{m_k}_k \subseteq U^{m_k}\) for all \(k \in K\). If all operations are total, then we say the structure is \emph{total}.
\end{definition}

\begin{definition}
Let \(I, J, K\) be finite indexing sets. Let \((U, \{f^{n_i}_i\}_{i \in I}, \{a_j\}_{j \in J}, \{R^{m_k}_k\}_{k \in K})\) be a partial structure and \((U', \{f'^{n_i}_i\}_{i \in I}, \{a'_j\}_{j \in J}, \{R'^{m_k}_k\}_{k \in K})\) be a total structure. Let \(\iota: U \to U'\) be an injection. We say \(\iota\) is an \emph{embedding}, if:
\begin{enumerate}[(i)]
\item for all \(j \in J\) we have \(\iota(a_j) = a'_j\),
\item for all \(i \in I\) and all \(x_1, x_2, \dots, x_{n_i} \in U\), if \(f^{n_i}_i (x_1, x_2, \dots, x_{n_i}) \neq \infty\),\\ then \(\iota(f^{n_i}_i (x_1, x_2, \dots, x_{n_i}) = f'^{n_i}_i (\iota(x_1), \iota(x_2), \dots, \iota(x_{n_i}))\),
\item for all \(k \in K\) we have \((\iota(x_1), \iota(x_2), \dots, \iota(x_{m_k})) \in (R'^{m_k}_k) \iff (x_1, x_2, \dots, x_{m_k}) \in R^{m_k}_k\)\\ for all \(x_1, x_2, \dots, x_{m_k} \in U\).
\end{enumerate}
\end{definition}
If \(\mathbf{A}\) is a partial structure, \(\mathbf{B}\) is a total structure and there exists an embedding from \(\mathbf{A}\) to \(\mathbf{B}\), then we say \(\mathbf{A}\) is \emph{embeddable} into \(\mathbf{B}\). If \(\mathbf{A}\) is embeddable into \(\mathbf{B}\) and \(A \subseteq B\), then we say \(\mathbf{A}\) is a \emph{partial substructure} of \(\mathbf{B}\). Let \(\mathcal{K}\) be a class of structures. By \(\mathcal{K}^P\) we denote the class of all partial substructures of structures of \(\mathcal{K}\).

\begin{definition}
Let \(\mathbf{L} = (L, \vee, \wedge, \top, \bot, \leq)\) be a partial structure. We say \(\mathbf{L}\) is a \emph{partial lattice}, if there exists a total lattice \(\mathbf{L'}\) such that \(\mathbf{L}\) is embeddable into it. If \(\mathbf{L'}\) is distributive, then \(\mathbf{L}\) is a \emph{partial distributive lattice}.
\end{definition}

One shows that a partial structure \((L, \vee, \wedge, \top, \bot, \leq)\) is a partial bounded lattice, if \((L, \leq)\) is a poset, \(\top\) and \(\bot\) are bounds of \(\leq\) and \(\vee, \wedge\) are compatible with \(\leq\), i.e. if \(a \vee b \neq \infty\), then \(a \vee b\) is the supremum of \(\{a, b\}\) with respect to \(\leq\) and if \(a \wedge b \neq \infty\), then \(a \wedge b\) is the infimum of \(\{a, b\}\) with respect to \(\leq\). See \cite{shkatov}.

\begin{definition}
Let \(\mathbf{B} = (B, \otimes, \multimap, \multimapinv, \vee, \wedge, \neg, 1, \top, \bot, \leq)\) be a partial structure. We say \(\mathbf{B}\) is a \emph{partial residuated Boolean algebra}, if there exists a total residuated Boolean algebra such that \(\mathbf{B}\) is embeddable into it and for all \(a \in B\) we have \(\neg a \neq \infty\), \(\neg a \in B\), \(a \vee \neg a = \top\) and \(a \wedge \neg a = \bot\). One notices that \((B, \otimes, \multimap, \multimapinv, \vee, \wedge, \top, \bot, \leq)\) is a partial bounded distributive residuated lattice.
\end{definition}

\subsection{Filters}

Let \((P, \leq)\) be a poset and let \(A \subseteq P\). We say \(A\) is an \emph{upset}, if for all \(a \in A\) and all \(b \in P\) such that \(a \leq b\) we have \(b \in A\). Analogously, \(A\) is a \emph{downset}, if for all \(a \in A\) and \(b \in P\) such that \(b \leq a\) we have \(b \in A\).

For every poset \((P, \leq)\) and every element \(a \in P\) we define:
\[
[a) = \{b \in P: a \leq b\} \qquad (a] = \{b \in P: b \leq a\}
\]
One notices \([a)\) is an upset and \((a]\) is a downset.

\begin{definition}
Let \((L, \vee, \wedge)\) be a lattice and let \(F \subseteq L\). We say \(F\) is a \emph{filter}, if the following conditions hold:
\begin{marray}
\text{(F1)} & \implies{a \leq b \land a \in F}{b \in F}\\
\text{(F2)} & \implies{a \in F \land b \in F}{a \wedge b \in F}
\end{marray}
We say \(F\) is \emph{proper}, if \(F \neq L\). The filter \(F\) is \emph{prime}, if it is proper and:
\begin{marray}
\text{(F3)} & \implies{a \vee b \in F}{a \in F \lor b \in F}
\end{marray}
\end{definition}

Let \((L, \vee, \wedge)\) be a lattice and \(F\) be a filter. We use the following notion:
\[ F_a = \left\{y \in L: \exists_{x \in F} x \wedge a \leq y\right\}\]
One proves \(F_a\) is a filter.

If we consider filters on residuated Boolean algebras, then (F3) is replaced with the following condition:
\begin{marray}
\text{(FB)} & \neg a \in F \text{ iff } a \notin F
\end{marray}

Considering filters on partial residuated Boolean algebras, we must change definition. We replace (F2) with the following condition:
\begin{marray}
\text{(F2')} & \implies{a \in F \land b \in F}{a \wedge b \in F \lor a \wedge b = \infty}
\end{marray}
for all \(a, b \in B\).

The following properties of filters are useful and may be easily proved.

\begin{lemma}
\label{bool-filter}
Let \((B, \vee, \wedge, \neg, \top, \bot)\) be a Boolean algebra and let \(F \subseteq B\) be a proper filter. The filter \(F\) is prime if, and only if, \(a \in F\) or \(\neg a \in F\) for all \(a \in B\).
\end{lemma}

This lemma remains true for residuated Boolean algebras.

\begin{proof}
Let \(F\) be a prime filter. Then \(a \vee \neg a = \top \in F\) for all \(a \in B\), so the condition of lemma holds. Now let \(a \in F\) or \(\neg a \in F\) for all \(a \in B\). Let \(a \vee b \in F\) and suppose \(a \notin F\) and \(b \notin F\). Then \(\neg a \in F\) and \(\neg b \in F\), by assumption. By (F2), \(\neg a \wedge \neg b \in F\). So, \(\neg(a \vee b) \in F\). Hence, \((a \vee b) \wedge \neg(a \vee b) = \bot \in F\), by (F2). This is impossible.
\end{proof}

\begin{lemma}
\label{sep-filter-strong}
Let \((L, \vee, \wedge)\) be a distributive lattice and let \(F \subseteq L\) be a filter and \(b \in L\) be such that \(b \notin F\). There exists a prime filter \(P \subseteq L\) such that \(F \subseteq P\) and \(b \notin P\).
\end{lemma}

\begin{proof}
Let \(F\) be a filter, \(b \in L\) and \(b \notin F\). We construct a prime filter as an extension of \(F\), but we need to avoid adding \(b\).

Let \(\mathcal{E}\) be a family of filters of \(L\) containing \(F\) and not containing \(b\). The family is nonempty, since \(F \in \mathcal{E}\). Let \(C \subseteq \mathcal{E}\) be any nonempty chain in \(\mathcal{E}\). Then \(F \subseteq \bigcup C\) and \(b \notin \bigcup C\). We show \(\bigcup C\) is a filter. Let \(c, d \in \bigcup C\), then \(c \in G\) and \(d \in G'\) for some \(G, G' \in C\). Since \(C\) is a chain, then \(G \subseteq G'\) or \(G' \subseteq G\), so both \(c\) and \(d\) are elements of \(G\) or \(G'\). Then, by (F2), \(c \wedge d \in G\) or \(c \wedge d \in G'\), so \(c \wedge d \in \bigcup C\). So \(\bigcup C\) satisfies (F2). (F1) is obvious. Hence, \(\bigcup C\) is a filter.

By Kuratowski--Zorn's lemma, there exists \(P \in \mathcal{E}\), which is a maximal element of \(\mathcal{E}\). We need to show \(P\) is prime. Let \(c, d \notin P\) and \(c \vee d \in P\). Since \(c \notin P\), then \(P \subseteq P_c\), and, since \(P\) is a maximal element of \(\mathcal{E}\), \(P_c \notin \mathcal{E}\). Clearly, \(F \subseteq P_c\), so \(b \in P_c\). Analogously, since \(d \notin P\), then \(b \in P_d\).

By definition of \(P_c, P_d\), for some \(x, y \in P\) we have \(x \wedge c \leq b\) and \(y \wedge d \leq b\). Hence, \(x \wedge y \wedge c \leq b\) and \(x \wedge y \wedge d \leq b\) and so \((x \wedge y \wedge c) \vee (x \wedge y \wedge d) \leq b\). By distributivity, \(x \wedge y \wedge (c \vee d) \leq b\). Since \(x, y, c \vee d \in P\), then \(b \in P\). Thus, if \(c, d \notin P\), when \(c \vee d \in P\), then \(b \in P\), which is impossible by definition of \(P\).
\end{proof}

\begin{corollary}
\label{sep-filter}
Let \((L, \vee, \wedge)\) be a distributive lattice and let \(a, b \in L\) be such that \(a \not\leq b\). There exists a prime filter \(F \subseteq L\) such that \(a \in F\) and \(b \notin F\).
\end{corollary}

\begin{proof}
The set \([a)\) is a filter such that \(b \notin [a)\). Then, by \Cref{sep-filter-strong}, there exists a prime filter \(P\) such that \(a \in P\) and \(b \notin P\).
\end{proof}

\begin{lemma}
\label{product1prime}
Let \(\mathbf{LB}\) be a total residuated Boolean algebra and let \(F, G\) be proper filters of \(\mathbf{B}\) and \(\mathbf{H}\) be a prime filter of \(\mathbf{H}\) such that \(\{ x \otimes y: x \in F \land y \in G \} \subseteq H\). Then, there exist prime filters \(F'\) and \(G'\) such that \(F \subseteq F'\) and \(G \subseteq G'\) and \(\{ x \otimes y: x \in F' \land y \in G \} \subseteq H\) and \(\{ x \otimes y: x \in F \land y \in G' \} \subseteq H\).
\end{lemma}
\begin{proof}
Let \(F, G\) be proper filters and \(H\) be a prime filter such that \(\{ x \otimes y: x \in F \land y \in G \} \subseteq H\). We show there exists a prime filter \(F'\) such that \(F \subseteq F'\) and \(\{ x \otimes y: x \in F' \land y \in G \} \subseteq H\).

Let \(\mathcal{E}\) be the family of filters \(Q\) of \(\mathbf{B}\) such that \(\{ x \otimes y: x \in Q \land y \in G \} \subseteq H\). This family is nonempty, since \(F \in \mathcal{E}\). Clearly, all filters in \(\mathcal{E}\) are proper; otherwise \(\bot = \bot \otimes 1 \in H\), which is impossible. We show that \(\bigcup C \in \mathcal{E}\) for every nonempty chain \(C \subseteq \mathcal{E}\). Now, let \(a \in \bigcup C\). Then, for some \(Q \in C\) we have \(a \in Q\) and \(\{ x \otimes y: x \in Q \land y \in G \} \subseteq H\). Hence, for some \(y \in G\), we have \(a \otimes y \in H\). So, \(\bigcup C \in \mathcal{E}\).

By Kuratowski--Zorn's lemma, there exists \(P \in \mathcal{E}\), which is a maximal element of \(\mathcal{E}\). We show \(P\) is a prime filter. Let \(a \vee b \in P\) and suppose \(a, b \notin P\). We consider \(P_a, P_b\). Clearly, \(P \subset P_a\) and \(P \subset P_b\). So, since \(P\) is a maximal element, \(P_a, P_b \notin \mathcal{E}\). So \(\{ x \otimes y: x \in P_a \land y \in G \} \not\subseteq H\) and \(\{ x \otimes y: x \in P_b \land y \in G \} \not\subseteq H\).

So, for some \(x, y \in P\) and some \(z_1, z_2 \in G\) we have \((x \wedge a) \otimes z_1 \notin H\) and \((y \wedge b) \otimes z_2 \notin H\). Since \(x, y, a \vee b \in P\), then \(x \wedge y \wedge (a \vee b) \in P\). So we have \((x \wedge y \wedge (a \vee b)) \otimes (z_1 \wedge z_2) \in H\). But:
\begin{multline*}
(x \wedge y \wedge (a \vee b)) \otimes (z_1 \wedge z_2) = ((x \wedge y \wedge a) \vee (x \wedge y \wedge b)) \otimes (z_1 \wedge z_2) =\\
=  (x \wedge y \wedge a) \otimes (z_1 \wedge z_2) \vee (x \wedge y \wedge b) \otimes (z_1 \wedge z_2)
\end{multline*}
So, since \(H\) is a prime filter, \((x \wedge y \wedge a) \otimes (z_1 \wedge z_2) \in H\) or \((x \wedge y \wedge b) \otimes (z_1 \wedge z_2) \in H\). Because \(H\) is a filter, then \((x \wedge a) \otimes z_1 \in H\) or \((y \wedge b) \otimes z_2 \in H\). This contradicts the assumptions. Hence, \(a \in P\) or \(b \in P\).

We put \(F' = P\). We show that there exists \(G'\) such that \(G \subseteq G'\) and \(\{ x \otimes y: x \in F \land y \in G' \} \subseteq H\) analogously.
\end{proof}

\begin{corollary}
\label{product2prime}
Let \(\mathbf{B}\) be a total residuated Boolean algebra and let \(F, G\) be proper filters of \(\mathbf{L}\) and \(\mathbf{H}\) be a prime filter of \(\mathbf{H}\) such that \(\{ x \otimes y: x \in F \land y \in G \} \subseteq H\). Then, there exist prime filters \(F'\) and \(G'\) such that \(F \subseteq F'\) and \(G \subseteq G'\) and \(\mathscr{R}_\mathbf{L}(F',G',H)\).
\end{corollary}
\begin{proof}
First, we construct \(F'\) such that \(\{ x \otimes y: x \in F' \land y \in G \} \subseteq H\), by \Cref{product1prime}. Then, we construct \(G'\) such that \(\{ x \otimes y: x \in F' \land y \in G' \} \subseteq H\), by \Cref{product1prime}. Then, by \Cref{r-eqv}, \(\mathscr{R}_\mathbf{L}(F',G',H)\).
\end{proof}

\subsection{Residuated frames}

\begin{definition}
Let \(\mathfrak{F} = (P, I, R)\). We say \(\mathfrak{F}\) is a \emph{residuated frame}, when \(I \subset P\) and \(R\) is a ternary relation on \(P\) and the following conditions hold:
\begin{marray}
\text{(U1)} & \forall_{x, x', y, z \in P} \big( \implies{R(x, y, z) \land x' = x}{R(x',y,z)}\big)\\
\text{(U2)} & \forall_{x, y, y', z \in P} \big( \implies{R(x, y, z) \land y' = y}{R(x,y',z)}\big)\\
\text{(U3)} & \forall_{x, y, z, z' \in P} \big( \implies{R(x, y, z) \land z = z'}{R(x,y,z')}\big)\\
\text{(U4)} & \forall_{x \in P}\exists_{y, z \in I} \big( R(x,y,x) \land R(z,x,x) \big)\\
\text{(U5)} & \forall_{x, z \in P}\forall_{y \in I} \big( \implies{R(x,y,z) \lor R(y,x,z)}{x = z} \big)
\end{marray}
\end{definition}

Residuated frames are the relational structures similar to groupoids. Instead of a binary operation we use a ternary relation. 

\begin{definition}
Let \(\mathbf{B} = (B, \otimes, \multimap, \multimapinv, \vee, \wedge, \neg, 1, \top, \bot, \leq)\) be a partial residuated Boolean algebra. We define the \emph{associated residuated frame} \(\mathfrak{F}_\mathbf{B} = (\mathscr{F}(B), \mathscr{I}_\mathbf{B}, \mathscr{R}_\mathbf{B})\), where \(\mathscr{F}(B)\) is the set of prime filters of \(\mathbf{B}\), \(\mathscr{I}_\mathbf{B}\) is the set of all prime filters containing 1 and:
\begin{align*}
\mathscr{R}_\mathbf{B}(F,G,H) \iff & \left( \forall_{a,b \in B} \implies{a \in F \land b \in G}{ a \otimes b \in H \vee a \otimes b = \infty} \right)\\
& \land \left(\forall_{a, b \in B} \implies{a \in F \land a \multimap b \in G \land a \multimap b \neq \infty }{ b \in H} \right)\\
& \land \left(\forall_{a, b \in B} \implies{b \multimapinv a \in F \land a \in G \land a \multimapinv b \neq \infty}{ b \in H} \right).
\end{align*}
\end{definition}

\begin{proposition}
\label{unital-filters}
Let \(\mathbf{B}\) be a residuated Boolean algebra and let \(F \in \mathscr{F}(B)\). Then, there exist prime filters \(P, Q \in \mathscr{F}(B)\) such that \(\mathscr{R}_\mathbf{B}(F,P,F)\) and \(\mathscr{R}_\mathbf{B}(Q,F,F)\) and \(1 \in P, 1 \in Q\).
\end{proposition}

\begin{proof}
Let \(F \in \mathscr{F}(L)\), we show there exists a prime filter \(P\) such that \(1 \in P\) and \(\mathscr{R}_\mathbf{L}(F,P,F)\). The proof for \(\mathscr{R}_\mathbf{L}(Q,F,F)\) is similar.

Let \(\mathcal{E}\) be the family of filters of \(\mathbf{L}\) such that for every filter \(G \in \mathcal{E}\) we have \(1 \in G\) and \(f \otimes g \in F\) for all \(f \in F\) and \(g \in G\). Clearly, all filters in \(\mathcal{E}\) are proper. This family is nonempty, since \([1) \in \mathcal{E}\). One shows that \(\bigcup C\) is a filter for every nonempty chain \(C \subseteq \mathcal{E}\) analogously like in the proof of \Cref{sep-filter-strong}. We show \(\bigcup C \in \mathcal{E}\). Clearly, \(1 \in \bigcup C\). Let \(f \in F\) and \(g \in \bigcup C\). Then, \(g \in G\) for some \(G \in C\). So, \(f \otimes g \in F\).

By Kuratowski--Zorn's lemma, there exists \(P \in \mathcal{E}\), which is a maximal element of \(\mathcal{E}\). We show that \(P\) is a prime filter. Assume \(a \vee b \in P\). Suppose \(a, b \notin P\).

We consider \(P_a\) and \(P_b\). Clearly, \(P \subset P_a\) and \(P \subset P_b\). Since \(P\) is a maximal element of \(\mathcal{E}\), then \(P_a, P_b \notin \mathcal{E}\).

We have \(1 \in P_a, P_b\). Then, for some \(f_a \in F\) and some \(x \in P\), we have \(f_a \otimes (x \wedge a) \notin F\) and for some \(f_b \in F\) and some \(y \in P\) we have \(f_b \otimes (y \wedge b) \notin F\). Since \(f_a, f_b \in F\), then \(f_a \wedge f_b \in F\), by (F2). Since \(a \vee b \in P\), then \((x \wedge y) \wedge (a \vee b) = (x \wedge y \wedge a) \vee (x \wedge y \wedge b) \in P\). 

So, \((f_a \wedge f_b) \otimes [(x \wedge a) \vee (y \wedge b)] \in F\). As a consequence:
\[
(f_a \wedge f_b) \otimes [(x \wedge a) \vee (y \wedge b)] = ((f_a \wedge f_b) \otimes (x \wedge a)) \vee ((f_a \wedge f_b) \otimes (y \wedge b))
\]
Because \(F\) is a prime filter, then \((f_a \wedge f_b) \otimes (x \wedge a) \in F\) or \((f_a \wedge f_b) \otimes (y \wedge b) \in F\). Assume \((f_a \wedge f_b) \otimes (x \wedge a) \in F\). Then \(f_a \otimes (x \wedge a) \in F\), by (F1) and monotonicity of \(\otimes\). Assume \((f_a \wedge f_b) \otimes (y \wedge b) \in F\). Then \(f_b \otimes (y \wedge b) \in F\). Both possibilites lead to the contradiction with assumptions. Hence, \(a \in P\) or \(b \in P\).

Therefore, \(\mathscr{R}_\mathbf{L}(F,P,F)\).
\end{proof}

\begin{lemma}
\label{r-eqv}
Let \(\mathbf{B}\) be a total residuated Boolean algebra and \(\mathfrak{F}_\mathbf{B} = (\mathscr{F}(B), \subseteq, \mathscr{R}_\mathbf{B})\) its associated residuated frame. Then, for \(F, G, H \in \mathscr{F}(B)\), the following are equivalent:
\begin{enumerate}[(i)]
\item \(\implies{a \in F \land b \in G}{ a \otimes b \in H}\) for all \(a, b \in B\)
\item \(\implies{a \in F \land a \multimap b \in G}{b \in H}\) for all \(a, b \in B\)
\item \(\implies{b \multimapinv a \in F \land a \in G }{b \in H}\) for all \(a, b \in B\)
\end{enumerate}
\end{lemma}
\begin{proof}
We assume (i). Let \(a \in F\) and \(a \multimap b \in G\). Since \(\mathscr{R}_\mathbf{B}(F,G,H)\), \(a \otimes (a \multimap b) \in H\) and then \(b \in H\), because \(a \otimes (a \multimap b) \leq b\). Hence (ii) holds. Now we assume (ii). Let \(a \in F\) and \(b \in G\). Since \(b \leq a \multimap (a \otimes b)\), then \(a \multimap (a \otimes b) \in G\), so, by (ii), \(a \otimes b \in H\) and (i) holds. The proof of equivalence of (i) and (iii) is similar.
\end{proof}

We construct a residuated Boolean algebras from the arbitrary residuated frame \(\mathfrak{F} = (P, I, R)\). Let \(X, Y \subseteq P\), we define:
\begin{gather*}
X \ocirc Y = \left\{ z \in P : \exists_{x, y \in P} x \in X \land y \in Y \land R(x, y, z) \right\}\\
X \omultimap Y = \left\{ y \in P:  \forall_{x, z \in P} \implies{R(x, y, z) \land x \in X }{z \in Y}\right\}\\
Y \omultimapinv X = \left\{ x \in P:  \forall_{y, z \in P}\implies{ R(x, y, z) \land y \in X}{z \in Y} \right\}
\end{gather*}
Then, \(\mathbf{B}_\mathfrak{F} = (\powerset{P}, \ocirc, \omultimap, \omultimapinv, \cup, \cap, {^c}, I, P, \emptyset, \subseteq)\) is a residuated Boolean algebra, where \(X^c = \powerset{P} \setminus X\) for all \(X \in \powerset{P}\). We call it the \emph{complex Boolean algebra of the residuated frame \(\mathfrak{F}\)}.

\begin{lemma}
\label{r-exist}
Let \(\mathbf{B}\) be a total residuated Boolean algebra and \(\mathfrak{F}_\mathbf{B} = (\mathscr{F}(B), \subseteq, \mathscr{R}_\mathbf{B})\) its associated residuated frame. Let \(a, b \in B\).
\begin{enumerate}[(1)]
\item If \(H \in \mathscr{F}(B)\) and \(a \otimes b \in H\), then there exist \(F, G \in \mathscr{F}(B)\) such that \(a \in F\), \(b \in G\) and \(\mathscr{R}_\mathbf{B}(F,G,H)\).
\item If \(G \in \mathscr{F}(B)\) and \(a \multimap b \not\in G\), then there exist \(F, H \in \mathscr{F}(B)\) such that \(a \in F\), \(b \not\in H\) and \(\mathscr{R}_\mathbf{B}(F,G,H)\).
\item If \(F \in \mathscr{F}(B)\) and \(b \multimapinv a \not\in F\), then there exist \(G, H \in \mathscr{F}(B)\) such that \(a \in G\), \(b \not\in H\) and \(\mathscr{R}_\mathbf{B}(F,G,H)\).
\end{enumerate}
\end{lemma}
\begin{proof}
We show (i). Since \(a \otimes b \in H\), then \(x \otimes y \in H\) for all \(a \leq x\) and \(b \leq y\). So, \(\{ x \otimes y: x \in [a) \land y \in [b)\} \subseteq H\) and, by \Cref{product2prime}, there exist prime filters \(F, G\) such that \(\mathscr{R}_\mathbf{B}(F,G,H)\).

We show (ii). Let \(G\) be a prime filter such that \(a \multimap b \notin G\). We consider \(aG = \{ a \otimes x : x \in G\}\). We extend \(aG\) to be filter. Let \(Q = \{x \in L: \exists_{y \in aG} y \leq x \}\). Clearly, (F1) holds. Let \(x, y \in Q\). Then, for some \(x', y' \in G\) we have \(a \otimes x' \leq x\) and \(a \otimes y' \leq y\). Since \(x', y' \in G\), then \(x' \wedge y' \in G\) and \(a \otimes (x' \wedge y') \in aG\). So:
\[
a \otimes (x' \wedge y') \leq (a \otimes x') \wedge (a \otimes y') \leq x \wedge y
\]
Hence, \(x \wedge y \in Q\). We show \(b \notin Q\). Suppose \(b \in Q\), then, for some \(x \in G\), \(a \otimes x \leq b\). By (RES), \(x \leq a \multimap b\). Hence, \(a \multimap b \in G\) -- contradiction. So, \(Q\) is a filter and \(b \notin Q\). By \Cref{sep-filter-strong}, there exists a prime filter \(H\) such that \(Q \subseteq H\) and \(b \notin H\). So, we have \(\{ x \otimes y : x \in [a) \land y \in G\} \subseteq H\). By \Cref{product1prime}, there exists a prime filter \(F\) such that \(\mathscr{R}_\mathbf{L}(F, G, H)\).

One shows (iii) analogously.
\end{proof}

\begin{lemma}
\label{sep-filter-partial}
Let \(\mathbf{B}\) be a partial residuated Boolean algebra and let \(a, b \in L\) be such that \(a \not\leq b\). There exists a prime filter \(F \subseteq B\) such that \(a \in F\) and \(b \notin F\).
\end{lemma}
\begin{proof}
By definition of a partial residuated Boolean algebra, there exists a total residuated Boolean algebra \(\mathbf{B'}\) such that \(\iota\) is an embedding of \(\mathbf{B}\) into \(\mathbf{B'}\). Then, by \Cref{sep-filter}, there exists a prime filter \(F \subseteq B'\) such that \(a \in F\) and \(b \notin F\). Clearly, \(\iota^{-1}(F)\) is a prime filter of \(\mathbf{B}\) and \(a \in \iota^{-1}(F)\) and \(b \notin \iota^{-1}(F)\).
\end{proof}

\begin{proposition}
Let \(\mathbf{B} = (B, \otimes, \multimap, \multimapinv, \vee, \wedge, \neg, 1, \top, \bot, \leq)\) be a partial residuated Boolean algebra. Let \(\mathbf{B}_{\mathfrak{F}_\mathbf{B}}\) be the complex Boolean algebra of the associated residuated frame. We define \(\iota(a) = \{ F \in \mathcal{F}_B : a \in F \}\) for all \(a \in B\). Then, \(\iota\) is an embedding.
\end{proposition}

\begin{proof}
Let \(a \leq b\). Then, for all \(H \in \iota(a)\), we have \(b \in H\), so \(H \in \iota(b)\). Hence, \(\iota(a) \subseteq \iota(b)\). Let \(a \not\leq b\). By \Cref{sep-filter-partial}, there exists a prime filter \(H\) such that \(a \in H\) and \(b \notin H\). Hence, \(\iota(a) \not\subseteq \iota(b)\). Therefore, \(a \leq b\) iff \(\iota(a) \subseteq \iota(b)\). As a consequence, \(\iota\) is injective.

Since prime filters are proper filters, \(\iota(\bot) = \emptyset\). \(\top\) is an element of every filter, so \(\iota(\top) = \mathscr{F}(B)\).

Let \(a, b \in B\) and \(a \otimes b \neq \infty\). By definition: 
\[
\iota(a) \ocirc \iota(b) = \left\{ H \in \mathscr{F}(B) : \exists_{F, G \in \mathscr{F}(B)} F \in \iota(a) \land G \in \iota(b) \land \mathscr{R}_\mathbf{B}(F, G, H) \right\}.
\]
We show \(\iota(a \otimes b) \subseteq \iota(a) \ocirc \iota(b)\). Let \(H \in \iota(a \otimes b)\). Then, \(a \otimes b \in H\) and by \Cref{r-exist}(i), there exist \(F, G \in \mathscr{F}(L)\) such that \(a \in F\), i.e. \(F \in \iota(a)\) and \(b \in G\), i.e. \(G \in \iota(b)\) and \(\mathscr{R}_\mathbf{B}(F,G,H)\).

We show \(\iota(a) \ocirc \iota(b) \subseteq \iota(a \otimes b)\). Let \(H \in \iota(a) \ocirc \iota(b)\). Then, for some \(F \in \iota(a)\) and \(G \in \iota(b)\) we have \(\mathscr{R}_\mathbf{B}(F,G,H)\). In particular, \(a \in F\), \(b \in G\), so \(a \otimes b \in H\), by definition of \(\mathscr{R}_\mathbf{B}\). Hence, \(H \in \iota(a \otimes b)\).

For \(a \multimap b\) and \(a \multimapinv b\) we prove analogously, using (ii) and (iii) of \Cref{r-exist} and \Cref{r-eqv}.

Let \(a \vee b \neq \infty\). We show \(\iota(a \vee b) \subseteq \iota(a) \cup \iota(b)\). Let \(H \in \iota(a \vee b)\), then \(a \vee b \in H\). Since \(H\) is a prime filter, \(a \in H\) or \(b \in H\). Hence, \(H \in \iota(a)\) or \(H \in \iota(b)\). Conversely, let \(a \in H\) or \(b \in H\). Then, \(a \vee b \in H\), by (F1). So, \(\iota(a) \cup \iota(b) \subseteq \iota(a \vee b)\).

Let \(a \wedge b \neq \infty\). Let \(H \in \iota(a \wedge b)\). Then, \(a \in H\) and \(b \in H\), by (F1). Hence, \(H \in \iota(a)\) and \(H \in \iota(b)\), i.e. \(H \in \iota(a)\). Conversely, let \(H \in \iota(a)\). Then, by (F2'), \(a \wedge b \in H\), so \(H \in \iota(a \wedge b)\).
\end{proof}

The following theorem allows us to identify the partial residuated Boolean algebras. Its proof is a merge of the proofs from \cite{shkatov} and \cite{vanalten}. We skip identical parts and we focus on nontrivial differences.

\begin{theorem}
\label{partial-B1}
Let \(\mathbf{B} = (B, \otimes, \multimap, \multimapinv, \vee, \wedge, \neg, 1, \top, \bot, \leq)\) be a partial structure such that \(\neg a \neq \infty\), \(\neg a \in B\), \(a \vee \neg a = \top\), \(a \wedge \neg a = \bot\) and \(1 \otimes a = a = a \otimes 1\) for all \(a \in B\). Then, \(\mathbf{B}\) is a partial unital residuated Boolean algebra if, and only if, it is a partial bounded lattice and there exists a set \(\mathcal{F}\) of prime filters of \(\mathbf{B}\) and a set \(\mathcal{I} \subseteq \mathcal{F}\) such that \(1 \in F\) for all \(F \in \mathcal{I}\) such that the following conditions hold:
\begin{marray}
\text{(S)} & \forall_{a, b \in L} \Big( \implies{a \not\leq b}{\exists_{F \in \mathcal{F}} a \in F \land b \not\in F} \Big)\\
\text{(M\(\otimes\))} & \forall_{H \in \mathcal{F}}\forall_{a,b \in L} \Big( \implies{a \otimes b \in H}{ \exists_{F,G \in \mathcal{F}} a \in F \land b \in G \land \mathscr{R}_\mathbf{L}(F,G,H)} \Big)\\
\text{(M\(\multimap\))} & \multicolumn{1}{l}{\forall_{G \in \mathcal{F}}\forall_{a,b \in L} \Big( \implies{a \multimap b \neq \infty \land a \multimap b \not\in G}{}\Big.}\\
& \multicolumn{1}{r}{\Big.\implies{}{\exists_{F,H \in \mathcal{F}} a \in F \land b \not\in H \land \mathscr{R}_\mathbf{L}(F,G,H)}\Big)}\\
\text{(M\(\multimapinv\))} & \multicolumn{1}{l}{\forall_{F \in \mathcal{F}}\forall_{a,b \in L} \Big(\implies{a \multimapinv b \neq \infty \land a \multimapinv b \not\in F}{}\Big.}\\
& \multicolumn{1}{r}{\Big.\implies{}{\exists_{G,H \in \mathcal{F}} a \in G \land b \not\in H \land \mathscr{R}_\mathbf{L}(F,G,H)}\Big)}\\
\text{(M1)} & \forall_{F \in \mathcal{F}} \exists_{G_1, G_2 \in \mathcal{I}} \Big( \mathscr{R}_\mathbf{L}(F,G_1,F) \land \mathscr{R}_\mathbf{L}(G_2, F, F) \Big)
\end{marray}
\end{theorem}
\begin{proof}
Let \(\mathbf{B} = (B, \otimes, \multimap, \multimapinv, \vee, \wedge, \neg, 1, \top, \bot, \leq)\) be a partial unital residuated Boolean algebra and let \(\mathbf{A} = (A, \ocirc, \omultimap, \omultimapinv, \ovee, \owedge, \oneg, \oi, \otop, \obot, \oleq)\) be a total unital residuated Boolean algebra and let \(\iota\) be an embedding of \(\mathbf{B}\) into \(\mathbf{A}\). We show that there exists a set \(\mathcal{F}\) of prime filters of \(\mathbf{B}\) that satisfies (S), (M\(\otimes\)), (M\(\multimap\)), (M\(\multimapinv\)) and (M1). We define:
\[
\mathcal{F} = \{\iota^{-1}(F): F\text{ is a prime filter of }\mathbf{A}\}
\]
For better readability we use the following notion: let \(F\) be a prime filter of \(\mathbf{A}\), then \(F_\iota = \iota^{-1}(F)\). We prove (S), (M\(\otimes\)), (M\(\multimap\)) and (M\(\multimapinv\)) like in \cite{shkatov}.

We show there exists \(\mathcal{I} \subseteq \mathcal{F}\) such that (M1) holds. We define:
\[
\mathcal{I} = \{F \in \mathcal{F}: 1 \in F \}
\]
Let \(F_\iota \in \mathcal{F}\), then, by \Cref{unital-filters} there exists a prime filter \(G\) of \(\mathbf{A}\) such that \(1 \in G\) and \(\mathscr{R}_\mathbf{A}(F,G,F)\). Then, \(G_\iota \in \mathcal{I}\) and \(\mathscr{R}_\mathbf{B}(F_\iota, G_\iota, F_\iota)\). Similarly, there exists \(H\) such that \(H_\iota \in \mathcal{I}\) and \(\mathscr{R}_\mathbf{B}(H_\iota, F_\iota, F_\iota)\).

Now we assume \(\mathbf{B}\) is a partial structure satisfying the assumptions of the theorem. We construct the residuated Boolean algebra \(\mathbf{A}\) and the embedding of \(\mathbf{B}\) into \(\mathbf{A}\). We see \(\mathfrak{F} = (\mathcal{F}, \mathcal{I}, \mathscr{R}_\mathbf{B})\) satisfies (U1)--(U4). We show (U5). Let \(F, H \in \mathcal{F}\) and \(G \in \mathcal{I}\) be such that \(\mathscr{R}_\mathbf{B}(F, G, H)\). Then, for all \(a \in F\), since \(1 \in G\), we have \(a \otimes 1 \in H\), so \(F \subseteq H\). Suppose there exists \(a \in H\) such that \(a \notin F\). Then, by (FB), \(\neg a \in F\), which is impossible.

Let \(\mathbf{A} = (\powerset{\mathcal{F}}, \otimes, \multimap, \multimapinv, \cup, \cap, \mathcal{I}, \mathcal{F}, \emptyset, \subseteq)\) be the complex algebra of \(\mathfrak{F}\). We define the mapping \(\iota\) for every \(a \in L\) by \(\iota(a) = \{ F \in \mathcal{F}: a \in F \}\). We show \(\iota\) is an embedding.

Let \(a, b \in L\) and \(a \leq b\). Then, \(\iota(a) \subseteq \iota(b)\), by (F1). Let \(a \not\leq b\), then by (S) there exists \(F \in \mathcal{F}\) such that \(a \in F\) and \(b \not\in F\), so \(\iota(a) \not\subseteq \iota(b)\). Hence \(a \leq b\) iff \(\iota(a) \subseteq \iota(b)\) and \(\iota\) is injective.

One shows \(\iota\) preserves \(\otimes, \multimap, \multimapinv, \vee, \wedge, \top, \bot\), analogously like in \cite{shkatov}.

We show \(\iota(1) = \mathcal{I}\). The inclusion \(\mathcal{I} \subseteq \iota(1)\) is trivial, since \(1\) belongs to every element of \(\mathcal{I}\). Let \(F \in \iota(1)\). By (M1), there exists \(G \in \mathcal{I}\) such that \(\mathscr{R}_\mathbf{B}(F,G,F)\). Since \(1 \in F\), then \(G \subseteq F\). Suppose \(a \in F\) and \(a \notin G\). Then, by (FB), \(\neg a \in G\) and then \(\neg a \in F\), which is impossible. So, \(G = F\) and \(F \in \mathcal{I}\).

Let \(a \in B\), then \(\iota(\neg a) = \{ F \in \mathcal{F}: \neg a \in F\} = \{ F \in \mathcal{F}: a \not\in F \}\), by (FB). Thus, \(\{ F \in \mathcal{F}: a \not\in F \} = \{ F \in \mathcal{F}: a \in F\}^c\).
\end{proof}

\section{The upper bound of complexity}

In this section we show that the finitary consequence relation for BFNL is decidable in exponential time.

\begin{lemma}
\label{partial-expB1}
Let \(\mathbf{B} = (B, \otimes, \multimap, \multimapinv, \vee, \wedge, \neg, 1, \top, \bot, \leq)\) be a partial structure. We can verify whether \(\mathbf{B}\) is a partial residuated Boolean algebra in exponential time (depending on \(|B|\)).
\end{lemma}

By definition, \(\mathbf{B}\) is a partial residuated Boolean algebra if it is embeddable in a total residuated Boolean algebra. Such a total algebra may have the same set of elements, but may also have additional elements to satisfy all the properties. Hence, to check if \(\mathbf{B}\) is a partial residuated Boolean algebra by definition, we need to embed \(\mathbf{B}\) in every possible total structure until we find one where all the properties of residuated Boolean algebra hold. Even with the limit on the maximal size of such a structure, it would be 2EXPTIME problem.

Hence, we use \Cref{partial-B1} to idenify partial residuated Boolean algebras.

\begin{proof}
We provide an algorithm to verify whether \(\mathbf{B}\) is a partial residuated Boolean algebra. We follow the analogous lemma and its proof from \cite{shkatov}.

\begin{enumerate}[Step 1.]
\item We check whether \(\leq\) is a partial order, \(\top, \bot\) are bounds and the lattice operators are compatible with \(\leq\). If it fails, the algorithm stops with negative answer. It can be done in the polynomial time.

\item We check whether \(1 \otimes a = a\) and \(a \otimes 1 = a\) for all \(a \in L\). If it fails, the algorithm stops with negative answer. It can be done in the polynomial time.

\item We check whether \(\neg a \neq \infty\), \(\neg a \in B\), \(a \vee \neg a = \top\) and \(a \wedge \neg a = \bot\) for all \(a \in B\). If it fails, the algorithm stops with negative answer. It can be done in the polynomial time.

\item We construct a descreasing sequence of families of filters \(\mathcal{F}_n\). We construct the set \(\mathcal{F}_0\) of all prime filters of \(\mathbf{B}\). For every subset \(S \subseteq B\) we check the definition of prime filter. It can be done in \(\mathcal{O}(2^{2|B|})\).

We set \(i = 0\).
\begin{enumerate}[Step \arabic{enumi}.1]
\item We define \(\mathcal{I}_i = \{ F \in \mathcal{F}_i: 1 \in F \}\). For every prime filter \(F \in \mathcal{F}_i\) we check (M\(\otimes\)), (M\(\multimap\)), (M\(\multimapinv\)) and (M1). If every of these condition holds for \(F\), then we add \(F\) to set \(\mathcal{F}_{i+1}\).

\item If \(\mathcal{F}_{i+1} = \emptyset\), then the algorithm stops with negative answer. If \(\mathcal{F}_i = \mathcal{F}_{i+1}\), then the algorithm proceeds to the next step. Else, the algorithm goes back to Step \arabic{enumi}.1 with \(i+1\).

\end{enumerate}

Checking conditions for arbitrary \(F\) can be done in \(\mathcal{O}(2^{3|B|})\). Number of filters in \(\mathcal{F}_i\) is \(\mathcal{O}(2^{|B|})\). Maximal \(i\) does not exceed \(2^{|B|}\). So this step can be done in \(\mathcal{O}(2^{5|B|})\).

\item We check (S). If (S) does not hold, then the algorithm stops with negative answer. If (S) does not hold for a family of filters, then it does not hold for any smaller family. It can be done in \(\mathcal{O}(|B|^2 2^{|B|})\) time.
\end{enumerate}

\end{proof}

We notice that every sequent \(\Gamma \Rightarrow C\) can be represented as \(G \Rightarrow C\), where \(G\) is a formula arising from \(\Gamma\) by replacing every comma by \(\otimes\), every semicolon by \(\wedge\), \(\epsilon\) by 1 and \(\delta\) by \(\top\). So, we consider only sequents of this form.

Let \(G \Rightarrow A\) be a sequent. We define the size of \(G \Rightarrow A\) as follows:
\begin{marray}
& s(p) = 1 & s(1) = 1\\
& s(\top) = 1 & s(\bot) = 1\\
& \multicolumn{2}{c}{s(A \otimes B) = s(A) + s(B) + 1} \\
& s(A \multimap B) = s(A) + s(B) + 1 & s(A \multimapinv B) = s(A) + s(B) + 1\\
& s(A \wedge B) = s(A) + s(B) + 1 & s(A \vee B) = s(A) + s(B) + 1\\
& s(\neg A) = s(A) + 1 & s(A \to B) = s(A) + s(B) + 1\\
& \multicolumn{2}{c}{s(G \Rightarrow A) = s(G) + s(A)}
\end{marray}

\begin{definition}
\label{partial-val}
Let \(\mathbf{A}\) be a partial residuated Boolean algebra. Let \(\mu\) be a partial function from the free algebra of \(\mathcal{L}\)--formulas into \(\mathbf{A}\). We say \(\mu\) is a \emph{valuation}, if the following conditions hold:
\begin{itemize}
\item \(\mu(\top) = \top\), \(\mu(\bot) = \bot\);
\item \(\mu(1) = 1\);
\item if \(\mu(D \otimes E) \neq \infty\), then \(\mu(D) \neq \infty, \mu(E) \neq \infty\) and \(\mu(D \otimes E) = \mu(D) \otimes \mu(E)\);
\item if \(\mu(D \multimap E) \neq \infty\), then \(\mu(D) \neq \infty, \mu(E) \neq \infty\) and \(\mu(D \multimap E) = \mu(D) \multimap \mu(E)\);
\item if \(\mu(D \multimapinv E) \neq \infty\), then \(\mu(D) \neq \infty, \mu(E) \neq \infty\) and \(\mu(D \multimapinv E) = \mu(D) \multimapinv \mu(E)\);
\item if \(\mu(D \wedge E) \neq \infty\), then \(\mu(D) \neq \infty, \mu(E) \neq \infty\) and \(\mu(D \wedge E) = \mu(D) \wedge \mu(E)\);
\item if \(\mu(D \vee E) \neq \infty\), then \(\mu(D) \neq \infty, \mu(E) \neq \infty\) and \(\mu(D \vee E) = \mu(D) \vee \mu(E)\);
\item if \(\mu(\neg D) \neq \infty\), then \(\mu(D) \neq \infty\) and \(\mu(\neg D) = \neg \mu(D)\);
\end{itemize}

Let \(G \Rightarrow C\) be a sequent and \(\mu\) be a valuation. We say \(G \Rightarrow C\) is satisfied under the valuation \(\mu\), if \(\mu(G) \neq \infty\), \(\mu(C) \neq \infty\) and \(\mu(G) \leq \mu(C)\).
\end{definition}

Now we are ready to prove the \EXPTIME complexity of of the consequence relations. The following theorem was formulated in \cite{shkatov} in algebraic terms of satisfiability of quantifier--free first--order formulas of the language of residuated distributive lattices.
\begin{theorem}
\label{EXPTIME}
The finitary consequence relation of BFNL is \EXPTIME.
\end{theorem}
\begin{proof}
\begin{enumerate}[(1)]
\item Let \(\mathcal{K}\) be the class of residuated Boolean algebras, \(\Phi = \{G_1 \Rightarrow C_1, G_2 \Rightarrow C_2, \dots, G_k \Rightarrow C_k\}\) be a set of sequents and \(G \Rightarrow C\) a sequent. Let: 
\[n := 2(s(G_1 \Rightarrow C_1) + s(G_2 \Rightarrow C_2) + \dots + s(G_k \Rightarrow C_k) + s(G \Rightarrow C)) + 4.\]
We show that \(\Phi\) entails \(G \Rightarrow C\), if, and only if, for all \(\mathbf{A} \in \mathcal{K}^P\) such that \(|A| \leq n\) and all valuations \(\mu\), if all sequents from \(\Phi\) are satisfied in \(\mathbf{A}\) under the valuation \(\mu\) and both \(\mu(G)\) and \(\mu(C)\) are defined, then \(G \Rightarrow C\) is satisfied in \(\mathbf{A}\) under the valuation \(\mu\).

\item[(1.1)] Let \(\mathbf{A} \in \mathcal{K}^P\), \(|A| \leq n\) and \(\mu\) be a valuation. Assume all sequents from \(\Phi\) are satisfied in \(\mathbf{A}\) under the valuation \(\mu\) and both \(\mu(G)\) and \(\mu(C)\) are defined, but \(G \Rightarrow C\) is not satisfied, i.e. \(\mu(G) \not\leq \mu(C)\). Then, for some \(\mathbf{A}' \in \mathcal{K}\), we have an embedding \(\iota\) of \(\mathbf{A}\) into \(\mathbf{A}'\). Then, \(\iota(\mu(G_i)) \oleq \iota(\mu(C_i))\) for all \(i = 1, \dots, k\) and \(\iota(\mu(G)) \onleq \iota(\mu(C))\) in \(\mathbf{A}'\). Hence, for the valuation \(\mu' = \iota \circ \mu\) all sequents from \(\Phi\) are satisfied, but \(G \Rightarrow C\) is not satisfied in \(\mathbf{A'}\). Thus, \(\Phi\) does not entail \(G \Rightarrow C\).

\item[(1.2)]  Now let \(G \Rightarrow C\) not be satisfied in \(\mathbf{A}' \in \mathcal{K}\) under the valuation \(\mu'\), but all sequents from \(\Phi\) be satisfied under \(\mu'\). We construct \(\mathbf{A} \in \mathcal{K}^P\).

First, we define \(T\) as the set consisting of \(1, \top, \bot\) and all subformulas of \(G_1, C_1, \dots, G_k, C_k, G, C\). We put \(A = \{ \mu'(D) : D \in T \} \cup \{ \oneg\mu'(D) : D \in T\}\). In effect, negation is a total operation, but doing this does not change final complexity. We define partial operations as follows:
\begin{itemize}
\item if \(D \in T\) and \(D = E \otimes F\), then \(\mu'(E) \otimes \mu'(F) := \mu'(E \otimes F)\);
\item if \(D \in T\) and \(D = E \multimap F\), then \(\mu'(E) \multimap \mu'(F) := \mu'(E \multimap F)\);
\item if \(D \in T\) and \(D = E \multimapinv F\), then \(\mu'(E) \multimapinv \mu'(F) := \mu'(E \multimapinv F)\);
\item if \(D \in T\) and \(D = E \vee F\), then \(\mu'(E) \vee \mu'(F) := \mu'(E \vee F)\);
\item if \(D \in T\) and \(D = E \wedge F\), then \(\mu'(E) \wedge \mu'(F) := \mu'(E \wedge F)\);
\end{itemize}

We define \(1 \otimes a := a\) and \(a \otimes 1 := a\) and \(\neg a := \oneg a\) and \(a \vee \neg a := \top\) and \(a \wedge \neg a := \bot\) for all \(a \in A\).

We also define \(\leq = {\oleq} \cap A^2\). By the construction, \(|A| \leq n\) and \(\mathbf{A} \in \mathcal{K}^P\). We define \(\mu = \mu'_{|T}\). Clearly, \(\mu\) satisfies the conditions of \Cref{partial-val} and \(\mu(G_i) \leq \mu(C_i)\) for \(i = 1, \dots, k\) and \(\mu(G) \not\leq \mu(C)\) and both \(\mu(G)\) and \(\mu(C)\) are defined.

\item Thus, to verify whether \(\Phi \vdash G \Rightarrow C\) we check whether \(G \Rightarrow C\) is satisfied in all \(\mathbf{A} \in \mathcal{K}^P\) under every valuation \(\mu\) such that \(|A| \leq n\) and all sequents from \(\Phi\) are satisfied in \(\mathbf{A}\) under \(\mu\) and both \(\mu(G)\) and \(\mu(C)\) are defined.

We construct all partial residuated Boolean algebras with cardinality not exceeding \(n\). Each such a structure can be encoded by matrices. Every binary operation and order is encoded by a matrix of size \(\mathcal{O}(n^2)\) and negation is encoded by matrix of size \(\mathcal{O}(n)\). Each entry in the matrix can take \(\mathcal{O}(n)\) values (including \(\infty\)). Hence, we have \(\mathcal{O}(2^{Ln^3})\) possibilities, where \(L\) is a positive integer. We check whether such a structure is a partial residuated Boolean algebra, using \Cref{partial-expB1}. This step can be done in \(\mathcal{O}(2^{Ln^3}2^{5n})\).

For a given residuated Boolean algebra \(\mathbf{A}\) the number of all possible valuations is \(\mathcal{O}(|A|^n)\). Checking if all sequents from \(\Phi\) and \(G \Rightarrow C\) are satisfied under the arbitrary valuation is \(\mathcal{O}(n)\). Hence, checking whether \(\Phi\) entails \(G \Rightarrow C\) in \(\mathbf{A}\) is \(\mathcal{O}(2^{n^3})\).

The time of the whole algorithm is \(\mathcal{O}(2^{Ln^3}2^{5n}2^{n^3}) = \mathcal{O}(2^{(L+1)n^3+5n})\). 
\end{enumerate}
\end{proof}

The analogous result for BFL (associative version of BFNL) does not hold. BFL is a strongly conservative extension of L and the consequence relation of L is undecidable \cite{buszkoL}.

If we exclude the constant 1 from BFNL, the result remains true. Moreover, for 1-free BFNL the lower bound of complexity of the consequence relation is also EXPTIME, since 1-free BFNL is a strongly conservative extension of 1-free DFNL which is EXPTIME-complete \cite{shkatov}. The lower bound of complexity for BFNL or DFNL with 1 remains an open problem.

\nocite{kozak}
\bibliographystyle{eptcs}
\bibliography{example}
\end{document}